\documentclass[a4paper]{amsart}
\usepackage{amsowariMT}
\pdfoutput=1
\usepackage{stmaryrd}
\usepackage{amsaddr_scauthor}
\usepackage[numbers]{natbib}
\usepackage{amsmath,mathrsfs,amssymb}		

\usepackage[OT1]{fontenc}
\usepackage[utf8]{inputenc}
\usepackage{mathptmx}
\DeclareSymbolFont{symbols}{OMS}{cmsy}{m}{n}

 \newcommand{\FC}{\mathcal{F}}

\newcommand{\MC}{\mathcal{M}} 
 \newcommand{\PC}{\mathcal{P}}
\newcommand{\QC}{\mathcal{Q}} 
 
\newcommand{\UC}{\mathcal{U}} 
 \newcommand{\XC}{\mathcal{X}}

 \newcommand{\FB}{\mathbb{F}}

 \newcommand{\NB}{\mathbb{N}}
 \newcommand{\PB}{\mathbb{P}}
 \newcommand{\RB}{\mathbb{R}}

 \newcommand{\Dh}{\hat{D}}

 \newcommand{\Ph}{\hat{P}}
\newcommand{\Qh}{\hat{Q}}

 \newcommand{\Xh}{\hat{X}}
 \newcommand{\Zh}{\hat{Z}}

\begin{document}

\title[Admissible Strategies in Robust Utility Maximization]{On
  Admissible Strategies in Robust Utility Maximization }

\author[K. Owari]{ Keita Owari}
\address{Graduate School of Economics, The University of Tokyo\\
  7-3-1 Hongo, Bunkyo-ku, Tokyo 113-0033, Japan }
\curraddr{}
\email{\href{mailto:keita.owari@gmail.com}{keita.owari@gmail.com}}
\thanks{\emph{Mathematics Subject Classification (2010).}
  60G44 $\cdot$ 60H05 $\cdot$ 91G10 $\cdot$ 91G80 \\
  \emph{JEL Classification.} C61 $\cdot$ G11 \\
  Forthcoming in \emph{Mathematics and Financial Economics}. 
  }%

\keywords{Robust utility maximization $\cdot$ Convex duality $\cdot$
  Supermartingale $\cdot$ Martingale measures}


\dedicatory{}


\begin{abstract}
  The existence of optimal strategy in robust utility maximization is
  addressed when the utility function is finite on the entire real
  line. A delicate problem in this case is to find a ``good
  definition'' of admissible strategies to admit an optimizer. Under
  certain assumptions, especially a kind of time-consistency property
  of the set $\PC$ of probabilities which describes the model
  uncertainty, we show that an optimal strategy is obtained in the
  class of those whose wealths are supermartingales under all local
  martingale measures having a finite generalized entropy with one of
  $P\in\PC$.
 
\end{abstract}
\maketitle

\section{Introduction}
\label{sec:Intro}

This paper analyzes a \emph{qualitative aspect} of the problem of
robust utility maximization.  Given a utility function $U$ and a set
$\PC$ of probabilities which describes the model uncertainty, the
basic problem of this paper is to maximize the \emph{robust utility
  functional}
\begin{align*}
  X\mapsto \inf_{P\in\PC}E_P[U(X)]
\end{align*}
over all terminal wealths $x+\theta\cdot S_T=x+\int_0^T\theta dS$ of
admissible strategies $\theta$, where $S$ is an underlying
semimartingale.  When $U$ is finite only on the positive half-line,
the duality theory for this problem in the spirit of
\citep{kramkov_schachermayer1999,
  kramkov_schachermayer2003:_neces_and_suffic_condit_in} has been
studied in both \emph{quantitative} and \emph{qualitative} aspects
(e.g. \citep{schied_wu05}, \citep{schied07},
\citep{follmer_gundel06}).  In the case of utility taking finite
values for all $x\in\RB$,
\citep{owari11:_dualit_robus_utilit_maxim_unboun} shows the key
duality, while \citep{follmer_gundel06} and \citep{owari11:_AME} give
a \emph{partial result} on the existence of optimal strategy which we
shall complete in this paper. See also \citep{follmer_schied_weber09}
for more comprehensive reference and the background of the robust
utility maximization problem.

A key subtlety intrinsic to the case of utility on $\RB$ is \emph{ the
  ``good definition'' of admissible strategies $\theta$}, which will
constitute the central theme of this paper. In this case, a universal
and conceptually natural definition of admissibility is that
$\theta\cdot S$ is uniformly bounded from below by some constant,
which completely determines the quantitative nature of the
problem. This class, however, typically fails to admit an
optimizer. On the other hand, if $U$ is $-\infty$ on $\RB_-$, the only
natural (non-redundant) definition of admissibility is that the
stochastic integral $\theta\cdot S$ is bounded from below by $-x$, and
an optimal strategy is indeed obtained in this class under certain
mild assumptions (see \citep{schied_wu05,schied07}).

In the classical case (i.e., $\PC=\{P\}$, say), the question of the
good definition of admissibility is closely analyzed by
\citep{schachermayer2003:_super_martin_proper_of_optim_portf_proces}
following the observation by \citep{delbaen_et_al02} and
\citep{kabanov_stricker02} in the case of exponential
utility. \citep{schachermayer2003:_super_martin_proper_of_optim_portf_proces}
shows that a ``good definition'' which yields us an optimal strategy
is that $\theta\cdot S$ is a supermartingale under all local
martingale measures $Q$ which has a ``finite entropy'' with the
physical probability $P$. We denote the class of such $\theta$ by
$\Theta_V(P)$ (see Section~\ref{sec:Formulation} for precise
definitions including the meaning of ``finite entropy'').  Note that
this class contains the usual admissible class, and the
supermartingale property is consistent to the ``No-Arbitrage
philosophy''. Thus $\Theta_V(P)$ is acceptably natural choice
\emph{when a single physical probability is specified}.

In the general robust case with $\PC$ containing (infinitely) many
elements, \citep{follmer_gundel06} (see also \citep{owari11:_AME} for
a slight generalization) provides a partial analogue of the above
result which states that, under certain stronger assumptions, an
optimal strategy is obtained in the class of $\theta$ with
$\theta\cdot S$ being a supermartingale under all local martingale
measures $Q$ having a finite entropy w.r.t. a certain element
$\Ph\in\PC$ called a \emph{least favorable measure}, i.e., in the
class $\Theta_V(\Ph)$. Here a dissatisfaction comes of course from the
dependence of admissibility on $\Ph$. In \emph{philosophy}, $\PC$ is
the set of \emph{candidates} of real world models, and we do not know
which one is true. Thus an ``admissible strategy'' should be
\emph{universally} admissible for all candidates $P\in\PC$. Also, the
least favorable probability $\Ph$ is a part of solution to the dual
problem of robust utility maximization, hence the class
$\Theta_V(\Ph)$ is not \emph{a priori} available.

In this view, a seemingly natural admissible class is
$\bigcap_{P\in\PC}\Theta_V(P)$ which is universal and contains all
$\theta$ whose stochastic integrals are bounded below. Thus our
central question in this paper is:
\begin{question}
  \label{Q:KeyQuestion}
  Does the class $\bigcap_{P\in\PC}\Theta_V(P)$ admit an optimal
  strategy?
\end{question}

The main result (Theorem~\ref{thm:MainThm1}) states that this is
indeed the case if (in addition to standard assumptions) the set $\PC$
of candidate models has a \emph{time-consistency} property. We proceed
as follows. The first step is to construct a so-called ``optimal
claim'' for the abstract version of robust utility maximization, from
which a candidate of optimal strategy $\hat\theta$ is derived through
a predictable representation argument. This part is mostly standard
excepting some technicality, but we give a slightly better description
of optimal claim. Note that the additional time-consistency assumption
is not required at this stage.  The crucial step is to verify the
supermartingale property of $\hat\theta\cdot S$ under all local
martingale measures $Q$ which has a finite entropy with \emph{some}
$P\in\PC$ \emph{but its entropy with $\Ph$ is infinite}. We shall do
this by a (slight surprisingly) simple trick.

\section{Formulation}
\label{sec:Formulation}

We fix a complete probability space $(\Omega,\FC,\PB)$ as well as a
filtration $\FB=(\FC_t)_{t\in[0,T]}$ satisfying \emph{the usual
  conditions}, where $T\in(0,\infty)$ is a fixed time horizon. Though
many probabilities on $(\Omega,\FC)$ will appear in the sequel, the
probability $\PB$ plays the role of reference probability, i.e., every
probabilistic notion is defined under $\PB$ unless other probability
is explicitly specified as $E_P[\cdot]$, $L^1(P)$ etc. In particular,
the underlying asset prices $S$ is a $d$-dimensional $\PB$-càdlàg
semimartingale, and we assume:
\begin{equation}
  \label{eq:LocBDD}\tag{\text{A1}}
  S\text{ is $\PB$-locally bounded}.
\end{equation}

Let $\PC$ be a set of probabilities $P\ll \PB$, which we can (and do)
embed into $L^1$ via the mapping $P\mapsto dP/d\PB$. In this sense, we
assume:
\begin{equation}
  \label{eq:PCompact} \tag{\text{A2}}
  \PC\text{ is convex and $\sigma(L^1,L^\infty)$-compact}.
\end{equation}

We work with a utility function $U:\RB\rightarrow\RB$ which we assume
\begin{equation}
  \label{eq:AsInada}\tag{\text{A3}}
  U\text{ is differentiable, strictly concave on $\RB$, and }U'(-\infty)=\infty, \, U'(\infty)=0,
\end{equation}
and satisfies the condition of \emph{reasonable asymptotic
  elasticity}:
\begin{equation}
  \label{eq:RAE}\tag{\text{A4}}
  \liminf_{x\rightarrow-\infty}\frac{xU'(x)}{U(x)}>1\text{ and }\limsup_{x\rightarrow \infty}\frac{xU'(x)}{U(x)}<1.
\end{equation}
The conjugate of utility function $U$ is denoted by $V$, i.e.,
\begin{equation}
  \label{eq:Conj1}
  V(y):=\sup_{x\in\RB}(U(x)-xy),\quad y \in\RB.
\end{equation}
The assumptions (\ref{eq:AsInada}) and (\ref{eq:RAE}) guarantee that
$V$ is a ``nice'' convex function (see \citep{frittelli_rosazza04},
\citep{owari2011:APFM} for details). Using this function, we introduce
a generalized entropy:
\begin{equation}
  \label{eq:IntroDiv1}
  V(\nu|P):=
  \begin{cases}
    E_P[V( d\nu/dP)]&\text{ if }\nu\ll P,\\
    +\infty&\text{ otherwise}
  \end{cases}
\end{equation}
for any positive finite measure $\nu\ll\PB$ and $P\in\PC$. When
$U(x)=1 -e^{-x}$ (exponential utility) and $Q$ is a probability with
$Q\ll P$, we have $V(Q|P)=E_Q[\log(dQ/dP)]$, i.e., the relative
entropy. Abusing the terminology, we still call the map
$V(\cdot|\cdot)$ the \emph{generalized entropy } associated to $V$. We
define also the \emph{robust generalized entropy} by
\begin{equation}
  \label{eq:RobEntropy}
  V(Q|\PC):=\inf_{P\in\PC}V(Q|P)<\infty.
\end{equation}

Let $\MC_{loc}$ be the set of all local martingale measures for $S$,
i.e., probabilities $Q\ll\PB$ under which $S$ is a local
martingale. We then set
\begin{align}
  \label{eq:MV1}
  \MC_V&:=\{Q\in\MC_{loc}:\, V(Q|\PC)<\infty\}.
\end{align}
Generically, for any set $\QC$ of probabilities $Q\ll\PB$, we denote
by $\QC^e$ the set of $Q\in\QC$ with $Q\sim\PB$.  We assume the
existence of \emph{equivalent } local martingale measure with finite
entropy in the following sense:
\begin{equation}
  \label{eq:ELMMV}\tag{\text{A5}}
  \MC_V^e:=\{Q\in\MC_V:\,Q\sim\PB\}\neq \emptyset.
\end{equation}
In particular, this implies the existence of
$(Q,P)\in\MC_V^e\times\PC$ such that $Q\sim P\sim \PB$ and
$V(Q|P)<\infty$. See \citep{owari11:_AME} for detail and other
consequences of these assumptions.

Let $L(S)$ be the totality of all $(S,\PB)$-integrable $d$-dimensional
predictable processes, $L_0(S):=\{\theta\in L(S): \theta_0=0\}$, and
we denote by $\theta\cdot S$ the stochastic integral of $\theta\in
L(S)$ w.r.t. $S$. See e.g., \citep{jacod80} or
\citep{jacod_shiryaev03} for more information. When the utility
function is finite on the entire real line, a conceptually natural
choice of $\Theta$ is
\begin{equation}
  \label{eq:ThetaBB}
  \Theta_{bb}:=\{\theta\in L_0(S):\,
  \theta\cdot S_0=0,\,\theta\cdot S\text{ is bounded from below}\}.
\end{equation}
Then the value function of the robust utility maximization problem is
given by
\begin{equation}
  \label{eq:ProbBasic1}
  u(x):=\sup_{\theta\in\Theta_{bb}}\inf_{P\in\PC}E_P[U(x+\theta\cdot S_T)],\quad x\in \RB.
\end{equation}
When we seek an optimal strategy, however, the class $\Theta_{bb}$ is
typically too small to admit an optimal strategy. We thus have to
enlarge the admissible class. Our choice is the following.
\begin{align}
  \label{eq:Theta_V1} \Theta_V&:=\{\theta\in L_0(S):\, \theta\cdot
  S\text{ is a $Q$-supermartingale, }\forall Q\in\MC_V\}.
\end{align}

\begin{rem}[Another equivalent formulation]
  \label{rem:EquivFormulation}
  We have defined the classes $\MC_V$ and $\Theta_V$ through the
  \emph{robust} generalized entropy $Q\mapsto V(Q|\PC)$. But the
  following equivalent formulation is sometimes useful for comparison.
  For each $P\in\PC$, we set
  \begin{align}
    \label{eq:MGmeasP}
    \MC_V(P)&:=\{Q\in\MC_{loc}:\, V( Q|P)<\infty\},\\
    \label{eq:ThetaVP1}
    \Theta_{V}(P)&:=\{\theta\in L_0(S):\, \theta\cdot S\text{ is a
      $Q$-supermartingale }\forall Q\in\MC_V(P)\}.
  \end{align}
  When a single $P\in\PC$ is fixed as the physical probability, the
  class $\Theta_V(P)$ is shown to be an appropriate domain of utility
  maximization in
  \citep{schachermayer2003:_super_martin_proper_of_optim_portf_proces}.
  Recalling (\ref{eq:RobEntropy}), our choices $\MC_V$ and $\Theta_V$
  are rewritten respectively as
  \begin{align*}
    \MC_V=\bigcup_{P\in\PC}\MC_V(P),\quad \Theta_V=\bigcap_{P\in\PC}\Theta_V(P).
  \end{align*}
  Thus our definition (\ref{eq:Theta_V1}) is consistent to what we
  wrote in introduction.

\end{rem}

Under the assumptions (\ref{eq:LocBDD}) -- (\ref{eq:ELMMV}), a duality
result (Theorem~2.3 of
\citep{owari11:_dualit_robus_utilit_maxim_unboun}) is applicable,
which states in our case that for any $\Theta$ with
$\Theta_{bb}\subset \Theta\subset \Theta_V$, we have
\begin{equation}
  \label{eq:Duality1}
  u(x)%
  =  \sup_{\theta\in\Theta}\inf_{P\in\PC}E_P[U(x+\theta\cdot
  S_T)]
  =\inf_{\lambda>0}\inf_{Q\in\MC_V}(V(\lambda Q|\PC)+\lambda x).
\end{equation}
In particular, the value function is unchanged if we replace
$\Theta_{bb}$ by the larger class $\Theta_V$.  Under the same
assumptions, the right hand side, the \emph{dual problem} of the
(\ref{eq:ProbBasic1}), admits a solution $(\hat\lambda,\Qh)\in
(0,\infty)\times\MC_V$, and the infimum $V(\hat\lambda
\Qh|\PC)=\inf_{P\in\PC}V(\hat\lambda \Qh|P)$ is attained by a
$\Ph\in\PC$ 
since $\PC$ is weakly compact, and $V(\cdot|\cdot)$ is lower
semicontinuous. Thus the right hand side of (\ref{eq:Duality1}) is
also written as $V(\hat\lambda \Qh|\Ph)$, and we call the triplet
$(\hat\lambda,\Qh,\Ph)$ a dual optimizer.

A way of proving (\ref{eq:Duality1}) and the existence of a solution
$(\hat\lambda,\Qh)$ is to closely analyze the robust utility
functional $X\mapsto \inf_{P\in\PC}E_P[U(X)]$ on $L^\infty$
characterizing $V(\cdot|\PC)$ as its conjugate. Then the duality and
the existence of $(\hat\lambda,\Qh)$ follow \emph{simultaneously} from
Fenchel's duality theorem. See
\citep{owari11:_dualit_robus_utilit_maxim_unboun} for
detail. Alternatively, one can separate the dual problem into the
minimization of $\lambda \mapsto \inf_{Q\in\MC_V}V(\lambda
Q|\PC)+\lambda x$ and of $Q\mapsto V(\lambda Q|\PC)$ for each
$\lambda$. For the latter problem, called the \emph{robust
  $f$-projection}, \citep{follmer_gundel06} proves the existence by
establishing a uniform integrability criterion in terms of
$V(\cdot|\PC)$ in the spirit of the de la Vallée-Poussin theorem.

In contrast to the standard utility maximization, neither the
uniqueness of $(\hat\lambda,\Qh)$ (hence of the triplet
$(\hat\lambda,\Qh,\Ph)$) nor the equivalence $\Qh\sim\PB$ hold in the
robust case, as the following trivial example illustrates:
\begin{ex}
  Suppose $\MC_{loc}^e\neq\emptyset$, and that $\MC_{loc}$ contains an
  element $Q_0$ which is not equivalent to $\PB$. Then we take $\PC$
  so that $Q_0\in\PC\subset\MC_{loc}$. In this case, $\hat\lambda$ is
  uniquely determined as the minimizer of $\lambda\mapsto
  V(\lambda)+\lambda x$. Then a triplet $(\hat\lambda,Q,P)$ is a dual
  optimizer if (and only if) $P=Q\in \PC\subset \MC_{loc}$. Indeed, by
  Jensen's inequality and the strict convexity of $V$, $V(\lambda
  Q|P)=E_P[V(\lambda dQ/dP)]\geq V(\lambda)$ whenever $Q\ll P$, and
  the ``equality'' holds if and only if $Q=P$. Hence
  $(\hat\lambda,\Qh,\Ph)$ is not unique, and $(\hat\lambda, Q_0,Q_0)$
  is a solution with $Q_0\not\sim\PB$.
\end{ex}

As for the equivalence, we still have $\Qh\sim\Ph$ whenever
$(\hat\lambda,\Qh,\Ph)$ is a dual optimizer (see \citep{owari11:_AME},
Theorem~2.7). Also, by an exhaustion argument, there exists a
\emph{maximal solution} $(\hat\lambda,\Qh,\Ph)$ in the sense that if
$(\lambda, Q,P)$ is another dual optimizer, then $P\ll \Ph$ (hence
$Q\ll\Qh$) and $\lambda dQ/dP=\hat\lambda d\Qh/d\Ph$, $P$-a.s., where
the density $d\Qh/d\Ph$ is defined $\PB$-a.s. in the sense of
\emph{Lebesgue decomposition}. In particular, if
$(\hat\lambda,\Qh,\Ph)$ and $(\tilde\lambda,\tilde Q,\tilde P)$ are
two maximal solution, then
\begin{equation}
  \label{eq:MaximalSol1}
  \tilde  \lambda d\tilde Q/d\tilde P=\hat\lambda d\Qh/d\Ph,\text{ $\PB$-a.s.,}
\end{equation}
See \citep[][Theorem~2.5 and Proposition~4.7]{owari11:_AME}.  This
uniqueness is still useful in our purpose.  Note finally that even
such a maximal $\Qh$ may fail to be equivalent to the reference
probability $\PB$.  See \citep[][Example~2.5]{schied_wu05} for a
counter example. In the sequel, we fix such a maximal dual optimizer,
and call $\Ph$ a \emph{least favorable measure}.

The duality (\ref{eq:Duality1}) completely characterizes the
\emph{quantitative nature} of the problem (\ref{eq:ProbBasic1}). But
our aim in this paper is to discuss the \emph{qualitative nature},
especially the existence of optimal strategy in $\Theta_V$. To do
this, assumptions (\ref{eq:LocBDD}) -- (\ref{eq:ELMMV}) are not
enough, and we assume additionally
\begin{equation}
  \label{eq:AsFinUtil}\tag{\text{A6}}
  \sup_{\theta\in\Theta_{bb}}E_P[U(\theta\cdot S_T)]<\infty,\quad \forall P\in\PC^e.
\end{equation}

\begin{rem}
  \label{rem:ConseqFinUtil}
  Several remarks on assumption (\ref{eq:AsFinUtil}) are in order.
  \begin{enumerate}
  \item This assumption is automatically satisfied if
    $U(\infty):=\sup_xU(x)<\infty$ as exponential utility, and in this
    case, $U(X)^+\in \bigcap_{P\in\PC}L^1(P)$ for \emph{any} random
    variable $X$.  Therefore, the robust utility functional $X\mapsto
    \inf_{P\in\PC}E_P[U(X)]$ is well-defined on $L^0$ as a
    $[-\infty,\infty)$-valued concave functional.
  \item If $U(\infty)=\infty$, \citep[][Th.~1.1 and
    Remark~1.2]{bellini_frittelli02} show under (\ref{eq:RAE}) that
    (\ref{eq:AsFinUtil}) is equivalent to:
    \begin{equation}
      \label{eq:FinUtilConseqEntropy}
      \forall P\in\PC^e,\, \exists Q\in\MC_V\text{ such that } V(Q|P)<\infty.
    \end{equation}

    This is further equivalent to saying that $v_P(y)<\infty$ for all
    $y>0$ and $P\in\PC^e$, where $v_P$ is the dual value function
    \begin{align*}
      v_P(y):=\inf_{Q\in\MC_V}V(y Q|P),\quad y>0.
    \end{align*}
  \item We could state (\ref{eq:AsFinUtil}) with the whole $\PC$
    rather than $\PC^e$. But for our purpose, (\ref{eq:AsFinUtil}) is
    enough.  Recall that (\ref{eq:ELMMV}) implies in particular
    $\PC^e\neq \emptyset$. If $\bar P\in \PC^e$, we have $(P+\bar
    P)/2\in\PC^e$ for all $P\in\PC$, and $\|X\|_{L^1((P+\bar
      P)/2)}=(\|X\|_{L^1(P)}+\|X\|_{L^1(\bar P)})/2\geq
    \|X\|_{L^1(P)}/2$. Hence we have, for instance,
    \begin{enumerate}
    \item $\bigcap_{P\in\PC}L^1(P)=\bigcap_{P\in\PC^e}L^1(P)$;
    \item if $(X^n)$ is bounded in $L^1(P)$ for all $P\in\PC^e$, then
      the same is true for all $P\in\PC$.
    \end{enumerate}
    In particular, (\ref{eq:AsFinUtil}) (hence
    (\ref{eq:FinUtilConseqEntropy})) guarantees even in the case
    $U(\infty)=\infty$ that
    \begin{equation}
      \label{eq:ConseqFinUtil2}
      X\in\bigcap_{Q\in\MC_V}L^1(Q)\,\Rightarrow\, U(X)^+\in\bigcap_{P\in\PC}L^1(P).
    \end{equation}
    In fact, if $V(Q|P)<\infty$ and $X\in L^1(Q)$, Young's inequality
    implies $U(X)\leq V(dQ/dP)+(dQ/dP)X\in L^1(P)$, and we can take
    such a $Q\in\MC_V$ by (\ref{eq:FinUtilConseqEntropy}) for all
    $P\in\PC^e$.

  \end{enumerate}

\end{rem}

\begin{rem}[Continuation of Remark~\ref{rem:EquivFormulation}]
  \label{rem:ContinAdmissible}
  We give a brief comparison of admissible classes considered in
  literature. In \citep{owari11:_AME}, the class $\Theta_V(\Ph)$ is
  used to discuss the existence of optimal strategy, while
  \citep{follmer_gundel06} considered (implicitly) a slightly smaller
  class:
  \begin{align}
    \label{eq:ClassMGFG}
    \MC_V^0(\Qh,\Ph)&:=\{Q\in\MC_{loc}:\, V(\alpha
    Q+(1-\alpha)\Qh|\Ph)<\infty,\, \exists \alpha\in (0,1)\},\\
  \Theta_V^0(\Qh,\Ph)&:=\Biggl\{\theta\in L_0(S):%
      \begin{matrix}
        \theta\cdot S\text{ is a $Q$-supermartingale, }\\
        \forall Q\in\MC_V^0(\Qh,\Ph)
      \end{matrix}
      \Biggr\}.
  \end{align}
  Note that $\Theta_V^0(\Qh,\Ph)\subset \Theta_V(\Ph)$ since
  $\MC_V(\Ph)\subset\MC_V^0(\Qh,\Ph)$, while if we set
  $\Theta_m(\Qh):=\{\theta\in L_0(S):\theta\cdot S\text{ is a
    $\Qh$-martingale}\}$,
  \begin{align*}
    \Theta_V\cap \Theta_m(\Qh)\subset
    \Theta_V(\Ph)\cap\Theta_m(\Qh)=\Theta_V^0(\Qh,\Ph)\cap
    \Theta_m(\Qh).
  \end{align*}
  Thus $\Theta_V(\Ph)$ and $\Theta_V^0(\Qh,\Ph)$ are essentially
  equivalent for the existence of optimal strategy (see
  Theorem~\ref{thm:MainThm1}). We just emphasize here that our class
  $\Theta_V$ depends neither on particular $P\in\PC$ nor $Q\in\MC_V$,
  while $\Theta_V(\Ph)$ and $\Theta_V^0(\Qh,\Ph)$ do.
\end{rem}

We conclude this section by recalling a stability property of a set of
probability measures, called \emph{m-stability}, which will be used in
Theorem~\ref{thm:MainThm1} below.

\begin{dfn}[\citep{MR2276899}, Definition~1]
  \label{dfn:M-stable}
  A set $\QC$ of probability measures is said to be \emph{m-stable}
  (multiplicatively stable) if for any $Q\in\QC$, $Q'\in\QC^e$ with
  the density processes $Z_t=(dQ/d\PB)|_{\FC_t}$ and
  $Z'_t=(dQ'/d\PB)|_{\FC_t}$, as well as any stopping time $\tau\leq
  T$, a new probability $\bar Q$ defined by $d\bar Q/d\PB:=Z_\tau
  (Z'_T/Z'_\tau)$ is an element of $\QC$.
\end{dfn}

This property is equivalent to the \emph{time-consistency} of the
corresponding \emph{dynamic coherent monetary utility function}
$\phi_\tau(X):=\essinf_{Q\in\QC}E_Q[X|\FC_\tau]$: for any $X, Y\in
L^\infty$ and stopping times $\sigma\leq \tau$, $\phi_\tau(X)\leq
\phi_\tau(Y)$ implies $\phi_\sigma(X)\leq \phi_\sigma(Y)$. This is
further equivalent (under (\ref{eq:AsInada})) to the time-consistency
of the dynamic robust utility functional
$\UC_\tau(X):=\essinf_{Q\in\QC}E_Q[U(X)|\FC_\tau]$. See
\citep[][Theorem~12]{MR2276899} for details and precise formulation.
Note that the set $\MC_{loc}$ of \emph{all} local martingale measures
is m-stable.

\section{Main Results}
\label{sec:Main}

We first state a result on a ``weak solution'' to the problem
(\ref{eq:ProbBasic1}), which yields a candidate of optimal
strategy. Let
\begin{equation}
  \label{eq:IndirectX}
 \textstyle \XC:=\left\{X\in L^0:\, X\in \bigcap_{Q\in\MC_V}L^1(Q),\, \sup_{Q\in\MC_V}E_Q[X]\leq 0\right\}.
\end{equation}
Note that $\theta\cdot S_T\in \XC$ if $\theta\in\Theta_V$, and $X\in
\XC$ implies $U(x+X)^+\in \bigcap_{P\in\PC}L^1(P)$ for any $x\in\RB$,
by (\ref{eq:AsFinUtil}) and Remark~\ref{rem:ConseqFinUtil}. Thus the
robust utility functional $X\mapsto \inf_{P\in\PC}E_P[U(x+X)]$ is
well-defined on $\XC$.  

\begin{thm}
  \label{thm:Indirect}
  Suppose (\ref{eq:LocBDD}) -- (\ref{eq:AsFinUtil}), and let $x\in\RB$
  and $(\hat\lambda,\Qh,\Ph)$ be a maximal dual optimizer. Then there
  exists an $\Xh\in\XC$ such that $U(x+\Xh)\in \bigcap_{P\in\PC}L^1(P)$
  and
  \begin{equation}
    \label{eq:OptClaim1}
      \begin{split}
        u(x)&=\sup_{X\in\XC}\inf_{P\in\PC}E_P[U(x+X)]=
        \inf_{P\in\PC}E_P[U(x+\Xh)],
      \end{split}
    \end{equation}
    where the infimum is attained by $\Ph$. Moreover, there exists an
    $(S,\Qh)$-integrable predictable process $\hat\theta$ with
    $\hat\theta_0=0$ such that $\hat\theta\cdot S$ is a
    $\Qh$-martingale (not only local) and
    \begin{equation}
      \label{eq:PredRep1}
      x+\Xh=-V'(\hat\lambda d\Qh/d\Ph)=x+\hat\theta\cdot S_T, \, \Qh\text{-a.s.}
    \end{equation}
    In particular, $\Xh$ is $\Qh$-a.s. unique, and $\hat\theta$ is
    unique in the sense that $\hat\theta\cdot S$ is unique up to
    $\Qh$-indistinguishability.
\end{thm}
The proof is given in Section~\ref{sec:OptimClaim}. The first equality
in (\ref{eq:OptClaim1}) states that the robust utility maximization
over \emph{terminal wealths} $x+\theta\cdot S_T$ is (quantitatively)
equivalent to the \emph{indirect utility maximization} :
\begin{align*}
  u_{\MC_V}(x)=\sup_{X\in\XC}\inf_{P\in\PC}E_P[U(x+X)],
\end{align*}
while the random variable $\Xh$ is the so-called \emph{optimal
  contingent claim}. Such arguments are quite standard in (non-robust)
utility maximization, and also in the robust case,
\citep[Theorem~3.11]{follmer_gundel06} shows a similar result: under
(\ref{eq:LocBDD}) -- (\ref{eq:ELMMV}), the assertions of
Theorem~\ref{thm:Indirect} hold true \emph{except that} the sets
$\MC_V$ (in the definition (\ref{eq:IndirectX})) and $\PC$ are
replaced by $\MC_V^0(\Qh,\Ph)$ and $\PC^0(\Qh,\Ph)$ defined
respectively by Remark~\ref{rem:ContinAdmissible} and
\begin{align*}
  \PC^0(\Qh,\Ph)&:=\{P\in\PC:\, V(\Qh|\alpha
  P+(1-\alpha)\Ph)<\infty,\exists \alpha\in (0,1)\}.
\end{align*}
Note that our finite utility assumption (\ref{eq:AsFinUtil}) is
automatic if $\PC$ is replaced by $\PC^0(\Qh,\Ph)$. Also, when
$U(\infty)<\infty$, the set $\PC^0(\Qh,\Ph)$ actually coincides with
the whole set $\PC$ (\citep{follmer_gundel06}, Remark~3.10). However,
$\MC_V^0(\Qh,\Ph)$ still depends on $(\Qh,\Ph)$ which is the
\emph{solution} to the dual problem, hence not \emph{a priori}
available. On the other hand, our formulation is universal, which is a
slight, but qualitatively crucial contribution.

Theorem~\ref{thm:Indirect} suggests that the ``strategy'' $\hat\theta$
is a candidate of optimal strategy. However, we still have to prove
that this strategy is indeed \emph{admissible}.
\begin{thm}
  \label{thm:MainThm1}
  In addition to (\ref{eq:LocBDD}) -- (\ref{eq:AsFinUtil}), we assume
  that $\Qh\sim\PB$ and $\PC$ is m-stable. Then $\hat\theta$ is
  $(S,\PB)$-integrable (hence $(S,P)$-integrable for all $P\in\PC$),
  and $\hat\theta\cdot S$ is a supermartingale under all $Q\in
  \MC_V$. In particular, $\hat\theta$ belongs to $\Theta_V$ and is an
  optimal strategy.
\end{thm}

The proof is given in Section~\ref{sec:UniSM}.  When $\PC=\{\PB\}$,
the question of \emph{uniform supermartingale property} of this type
goes back to the ``six-author paper'' \citep{delbaen_et_al02} which
shows that the optimal wealth in \emph{exponential utility}
maximization is a \emph{martingale} under all local martingale
measures having a finite relative entropy with $\PB$, under an
additional assumption on reverse Hölder inequality which is later
removed by \citep{kabanov_stricker02}. Although this uniform
martingale property is no longer true for other utility
functions, 
\citep{schachermayer2003:_super_martin_proper_of_optim_portf_proces}
shows that the optimal wealth is a supermartingale under all
$Q\in\MC_V(\PB)$, for any utility functions on $\RB$ with reasonable
asymptotic elasticity. There are also some extensions to the case
where the semimartingale $S$ is not locally bounded. See
e.g. \citep{MR2295831} and \citep{BiaginiCerny2011}.

In the robust case, the $Q$-supermartingale property for all
$Q\in\MC_V(\Ph)$ (hence all $Q\in\MC_V^0(\Qh,\Ph)$ since
$\hat\theta\cdot S$ is a $\Qh$-martingale) is shown by
\citep{follmer_gundel06} (see also \citep{owari11:_AME} for a slight
extension). We emphasize that the difference between $\MC_V(\Ph)$ and
$\MC_V$ is essential here.  Note that $\Xh$ is also optimal for the
utility maximization problem under the fixed measure $\Ph$, and the
same is true for $(\hat\lambda, \Qh)$ in the dual side. Thus the
result of
\citep{schachermayer2003:_super_martin_proper_of_optim_portf_proces}
cited in the previous paragraph is still applicable (under the
assumption $\Qh\sim\PB$) for $Q$ with $V(Q|\Ph)<\infty$, while we have
to consider the case where $V(Q|P)<\infty$ for \emph{some} $P\in\PC$
\emph{but possibly $V(Q|\Ph)=\infty$}.

To grasp the situation, we try to describe the heuristics behind the
argument in
\citep{schachermayer2003:_super_martin_proper_of_optim_portf_proces}
(from our point of view), and our idea of extending it. In what
follows in this section, we suppose all the assumptions of
Theorem~\ref{thm:MainThm1}, especially $\Qh\sim\PB$.

For a moment, we \emph{suppose} 
that $\hat\theta\cdot S$ is a $Q$-supermartingale for some
$Q\in\MC_V$. Then the $\Qh$-martingale property and the representation
(\ref{eq:PredRep1}) imply: for any stopping time $\tau\leq T$,
\begin{equation}
  \label{eq:DVI1}
  E_{\Qh}[V'(\hat\lambda d\Qh/d\Ph)|\FC_\tau]\leq E_Q[V'(\hat\lambda d\Qh/d\Ph)|\FC_\tau],\,Q\text{-a.s.}
\end{equation}
On the other hand, Ansel-Stricker's lemma \citep{ansel_stricker94}
shows that $\hat\theta\cdot S$ is a $Q$-supermartingale if and only if
there exists a \emph{$Q$-martingale lower bound}, i.e., a
$Q$-martingale $M^Q$ such that $\hat\theta\cdot S\geq M^Q$,
$Q$-a.s. In particular, if (\ref{eq:DVI1}) holds true for any stopping
time $\tau\leq T$, the process defined by
$M^Q_\tau=-E_Q[V'(\hat\lambda d\Qh/d\Ph)|\FC_\tau]$ provides a desired
lower bound, hence (\ref{eq:DVI1}) is a necessary and
\emph{sufficient} condition for $\hat\theta\cdot S$ to be a
$Q$-supermartingale.

When $V(Q|\Ph)<\infty$, the inequality (\ref{eq:DVI1}) is obtained as
the \emph{variational inequality} which characterizes $\Qh$ as a
minimizer of the functional $Q\mapsto V(\hat\lambda Q|\Ph)$ when
$\tau=0$, and a ``Bellman-type'' principle using the m-stability of
the set of local martingale measures shows the case of general
$\tau\leq T$.

If $\inf_{P\in\PC}V(Q|P)<\infty$ but $V(Q|\Ph)=\infty$, this argument
is no longer applicable at least directly. Mathematically speaking, we
loose some important estimates to guarantee the necessary
convergences, or more intuitively, any element $Q$ with
$V(Q|\Ph)=\infty$ is in no way optimal at very early stage, and we can
not draw further information from the optimality of $\Qh$ in the
minimization of $Q\mapsto V(\hat\lambda Q|\Ph)$. However, we have used
only a part of information of $\Qh$ so far, and it is natural to
expect that a better information may improve the result.  More
specifically,
\begin{description}[\bf Step~2]
\item[\bf Step~1] the optimality of $(\Qh,\Ph)$ in the minimization of
  $(Q,P)\mapsto V(\hat\lambda Q|P)$ should yield a variational
  inequality similar to (\ref{eq:DVI1}) but with an additional term
  involving $P$:
  \begin{align*}
    \text{`` }E_{\Qh}[V'(\hat\lambda
    d\Qh/d\Ph)|\FC_\tau]+F_\tau(\Ph)\leq E_{Q}[V'(\hat\lambda
    d\Qh/d\Ph)|\FC_\tau]+F_\tau(P) \text{''}.
  \end{align*}
\item[\bf Step~2] Though we may not take $P=\Ph$ in general, it seems
  natural to expect that we may take $P$ ``arbitrarily close to
  $\Ph$'' keeping $V(Q|P)<\infty$ with fixed $Q$.
\item[\bf Step~3] If this is the case, we may expect (\ref{eq:DVI1})
  by an approximation argument:
  \begin{align*}
    \text{`` }F_\tau(P)\rightarrow F_\tau(\Ph)\text{''}.
  \end{align*}

\end{description}
The formal inequality in \textbf{Step 1} will be realized as
Proposition~\ref{prop:AuxVI} below, where the m-stability of $\PC$
will play an important role. On the other hand, \textbf{Steps~2} and
\textbf{3} will be justified in a certain sense by a simple trick
which is a consequence of reasonable asymptotic elasticity
(Lemma~\ref{lem:ConseqRAE}).

\begin{rem}[What happens when $\Qh\not\sim\PB$?]
  \label{rem:WhatIf}
  The equivalence $\Qh\sim\PB$ is automatic if all elements of $\PC$
  are equivalent to $\PB$. When the filtration $\FB$ is
  \emph{continuous} (i.e., every $(\FB,\PB)$-martingale is continuous,
  especially if it is generated by a Brownian motion), the latter
  condition is already implied by the m-stability of $\PC$ and
  (\ref{eq:PCompact}) (see \citep[][Theorem~8]{MR2276899}), thus it is
  not a further restriction in that case.


  In general, however, the equivalence $\Qh\sim\PB$ may fail (see
  \citep[][Example~2.5]{schied_wu05} for a counter example), thus it
  is worth asking what happens in that case. 
  When $U$ is finite only on the positive half-line, the optimal claim
  $\Xh$ (which does not require the assumption $\Qh\sim\PB$) is
  super-hedged by some $(S,\PB)$-integrable process $\tilde\theta$
  with $\tilde\theta\cdot S=\hat\theta\cdot S$, $\Qh$-a.s. By the
  monotonicity of robust utility functional, we see that
  $\tilde\theta$ is an optimal strategy without the additional
  assumption $\Qh\sim\PB$ (see \citep{schied_wu05} and
  \citep{schied07}). However, this argument essentially relies on the
  fact that $\Xh$ is bounded below by $-x$ (since $U(x)=-\infty$ for
  $x<0$), and no longer works when the utility function is finite on
  the entire real line. Thus we can not drop the assumption
  $\Qh\sim\PB$ (at now).

\end{rem}

\begin{rem}[Random Endowment]
  \label{rem:Endowment}
  The results of this paper may also be stated with a \emph{random
    endowment} $B$ as long as it is an $\FC_T$-measurable random
  variable satisfying
  \begin{align}
    \label{eq:Bpositive}
    \begin{split}
      &  \forall P\in\PC,\, \exists \varepsilon_P>0\text{ such that } U(-\varepsilon_PB^+)\in L^1(P),\\
      &\exists \varepsilon>0\text{ such that
      }\{U(-(1+\varepsilon)B^-)dP/d\PB\}_{P\in\PC}\text{ is uniformly
        integrable.}
    \end{split}
  \end{align}
  Then the robust utility maximization problem (\ref{eq:ProbBasic1})
  reads as
  \begin{equation}
    \label{eq:ProbEndowment}
    u_B(x):=\sup_{\theta\in\Theta_{bb}}\inf_{P\in\PC}E_P[U(x+\theta\cdot S_T+B)],
  \end{equation}
  Assumption (\ref{eq:Bpositive}) implies that $B\in
  \bigcap_{Q\in\MC_V}L^1(Q)$, and guarantees under (\ref{eq:LocBDD})
  -- (\ref{eq:ELMMV}) that a duality corresponding to
  (\ref{eq:Duality1}) holds true
  \citep[][Theorem~2.3]{owari11:_dualit_robus_utilit_maxim_unboun}:
  \begin{equation}
    \label{eq:DualityEndowment}
    \sup_{\theta\in\Theta_V}\inf_{P\in\PC}E_P[U(x+\theta\cdot
    S_T+B)]
    =\inf_{\lambda>0}\inf_{Q\in \MC_V}(V(\lambda
    Q|\PC)+\lambda x+\lambda  E_Q[B]),
  \end{equation}
  With the same assumptions, the dual problem admits a maximal
  solution with the unique density in the sense of
  (\ref{eq:MaximalSol1}). Then Theorems~\ref{thm:Indirect}
  and~\ref{thm:MainThm1} remain true with similar proofs, and with
  obvious modifications, e.g., (\ref{eq:PredRep1}) is replaced by $x+
  \Xh+B=-V'(\hat\lambda d\Qh/d\Ph)=x+\hat\theta\cdot S_T+B$,
  $\Qh$-a.s. We omit the details. See
  \citep{owari11:_dualit_robus_utilit_maxim_unboun} for the treatment
  of random endowment and other implications of (\ref{eq:Bpositive}).
\end{rem}

\section{Optimal Claim}
\label{sec:OptimClaim}

We first note that we have only to consider the case $x=0$. Indeed,
assumptions (\ref{eq:AsInada}) and (\ref{eq:RAE}) on the utility
function are invariant under the translation of utility function from
$U$ to $U_{x}(\xi):=U(x+\xi)$, and all the results for $x\neq 0$
follow from those for $x=0$ applied to the new utility function $U_x$.
Thus we assume $x=0$ in what follows.

The next technical lemma is a collection of several arguments in
\citep{BiaginiCerny2011}.
\begin{lem}[{\citep{BiaginiCerny2011}}]
  \label{lem:BCmod1}
  Let $(Q,P)$ be a pair of probabilities with $V(Q|P)<\infty$, and
  $(k^n)_n$ a sequence of random variables such that $E_P[U(k^n)]$ is
  bounded from below and $E_Q[k^n]\leq 0$ for all $n$. Then
  \begin{description}[(a)]
  \item[(a)] $(k^n)_n$ is bounded in $L^1(Q)$;
  \item[(b)] $(U(k^n))_n$ is bounded in $L^1(P)$;
  \item[(c)] If in addition $k^n$ converges a.s. to some $k\in L^0$,
    we have $k\in L^1(Q)$, $U(k)\in L^1(P)$ and that
  \begin{equation}
    \label{eq:BClimsup}
    E_Q[k]\leq 0\text{ and }    \limsup_nE_P[U(k^n)]\leq E_P[U(k)].
  \end{equation}

  \end{description}

\end{lem}
\begin{proof}
  We just fill the gap from \citep{BiaginiCerny2011}. As we are
  assuming the reasonable asymptotic elasticity (\ref{eq:RAE}),
  assertions (a) and (b) are contained in Proposition~6.3 of
  \citep{BiaginiCerny2011}. The assertion (c) also appears
  (implicitly) in the proof of their Theorem~4.10, which we briefly
  recall here.

  Assume $k^n\rightarrow k$, $P$-a.s. Since $(k^n)$ (resp. $(U(k^n)$)
  is bounded in $L^1(Q)$ (resp. $L^1(P)$), Fatou's lemma applied to
  the sequence $(|k^n|)_n$ (resp. $(|U(k^n)|)_n)$ shows that $k\in
  L^1(Q)$ (resp. $U(k)\in L^1(P)$). By Young's inequality, we have
  $U(k^n)-\lambda (dQ/dP)k^n\leq V(\lambda dQ/dP)\in L^1(P)$ for all
  $n\in\NB$ and $\lambda>0$, where the $P$-integrability of the right hand
  side for all $\lambda$ follows from the reasonable asymptotic
  elasticity. By this \emph{integrable upper bound} as well as the
  assumption $E_Q[k^n]\leq 0$, (reverse) Fatou's lemma shows that
  \begin{align*}
    \limsup_nE_P[U(k^n)]&\leq \limsup_nE_P[U(k^n)-\lambda
    (dQ/dP)k^n]\\
    &\leq E_P[U(k)-\lambda (dQ/dP)k]=E_P[U(k)]-\lambda E_Q[k], \,
    \forall \lambda>0.
  \end{align*}
  Letting $\lambda\downarrow 0$, we have (\ref{eq:BClimsup}), while
  $E_Q[k]\leq 0$ follows by letting $\lambda \uparrow \infty$.
\end{proof}

\begin{proof}[Proof of Theorem~\ref{thm:Indirect}.]
  We choose a maximizing sequence $(\theta^n)_n\subset \Theta_{bb}$,
  that is
  \begin{equation}
    \label{eq:MaxSeq1}
    \inf_{P\in\PC}E_P[U(\theta^n\cdot S_T)]\nearrow u(0).
  \end{equation}
  This sequence does not have to converge, thus we appeal to a Komlós
  type argument. Let $(\bar Q,\bar P)\in \MC_V\times\PC$ be such that
  $\bar Q\sim\bar P\sim \PB$ and $V(\bar Q|\bar P)<\infty$ which
  exists by (\ref{eq:ELMMV}). Since $E_{\bar P}[U(\theta^n\cdot
  S_T)]\geq \inf_{P\in\PC}E_P[U(\theta^1\cdot S_T)]$ and $E_{\bar
    Q}[\theta^n\cdot S_T]\leq 0$ by construction,
  Lemma~\ref{lem:BCmod1}~(a) shows that $(\theta^n\cdot S_T)_n$ is
  bounded in $L^1(\bar Q)$. Hence Komlós' theorem (see
  e.g. \citep[][Theorem~15.1.3]{delbaen_schachermayer2006:_mathem_of_arbit})
  yields another sequence $(\tilde k^n)_n$ such that
  \begin{align*}
    \begin{cases}
      \tilde k^n \in \mathrm{conv}(\theta^n\cdot S_T,\theta^{n+1}\cdot
      S_T,\cdots\,)\\
      \tilde k^n \text{ converges $\bar Q$-a.s. (hence $\PB$-a.s.)  to
        some }\Xh\in L^1(\bar Q).
    \end{cases}
  \end{align*}
  By construction, each $\tilde k^n$ is again the terminal value of a
  stochastic integral $\tilde\theta^n\cdot S_T$ where $\tilde
  \theta^n$ is the convex combination of
  $(\theta^n,\theta^{n+1},\cdots\,)$ with the same convex weights as
  $\tilde k^n$, hence $\tilde\theta^n\in \Theta_{bb}$ and $E_Q[\tilde
  k^n]\leq 0$ for each $n$ and $Q$, in particular.

  Since the robust utility functional $X\mapsto
  \inf_{P\in\PC}E_P[U(X)]$ is concave as a point-wise infimum of
  concave functionals, we have $\inf_{P\in\PC}E_P[U(\tilde k^n)]\geq
  \inf_{P\in\PC}E_P[U(\theta^n\cdot S_T)]$ for each $n$. Hence we
  still have $\lim_n\inf_{P\in\PC}E_P[U(\tilde k^n)]=u(0)$, and the
  sequence $(E_P[U(\tilde k^n)])_n$ is bounded from below for all
  $P\in\PC$.

  If $Q\in\MC_V$, there is a $P\in\PC$ with $V(Q|P)<\infty$ by the
  definition of $\MC_V$, hence another application of
  Lemma~\ref{lem:BCmod1} to the sequence $(\tilde k^n)$ with the pair
  $(Q,P)$ shows that $\Xh\in L^1(Q)$ and $E_Q[\Xh]\leq 0$. Hence
  $\Xh\in \XC$.

  We next show that $U(\Xh)\in \bigcap_{P\in\PC}L^1(P)$ and
  \begin{equation}
    \label{eq:ProofAlimsup}
    \limsup_nE_P[U(\tilde k^n)]\leq E_P[U(\Xh)],\quad \forall P\in\PC.
  \end{equation}
  This is immediate from Fatou's lemma if $U$ is bounded from
  above. When $U(\infty)=\infty$ and $P\in\PC^e$, we can take a
  $Q\in\MC_V$ with $V(Q|P)<\infty$ by (\ref{eq:FinUtilConseqEntropy}),
  hence Lemma~\ref{lem:BCmod1} shows (\ref{eq:ProofAlimsup}) and that
  $(U(\tilde k^n))_n$ is bounded in $L^1(P)$. Then
  Remark~\ref{rem:ConseqFinUtil} shows that $U(\Xh)\in L^1(P)$ and
  $(U(\tilde k^n))_n$ is still bounded in $L^1(P)$ for arbitrary
  $P\in\PC$ which need not be equivalent to $\PB$. To prove
  (\ref{eq:ProofAlimsup}) in the case $P\not\sim\PB$, we take $(\bar
  Q,\bar P)$ as above, and set $P_\alpha:=\alpha P+(1-\alpha)\bar P$
  for $\alpha\in(0,1)$. Since $P_\alpha\sim \PB$, the claim is true
  for $P_\alpha$ for all $\alpha \in (0,1)$, while we see that
  $\sup_n|E_{P_\alpha}[U(\tilde k^n)]-E_P[U(\tilde k^n)]|\leq
  2(1-\alpha) \sup_n(\|U(\tilde k^n)\|_{L^1(P)}\vee \|U(\tilde
  k^n)\|_{L^1(\bar P)})\rightarrow 0$, as $\alpha \uparrow 1$. Thus we
  deduce
  \begin{align*}
    \limsup_nE_P[U(\tilde k^n)] &=\lim_{\alpha\uparrow 1}\limsup_n
    E_{\alpha P+(1-\alpha)\bar P}[U(\tilde k^n)]\\
    &\leq \lim_{\alpha\uparrow 1}E_{\alpha P+(1-\alpha)\bar
      P}[U(\Xh)]=E_P[U(\Xh)].
  \end{align*}
  Hence (\ref{eq:ProofAlimsup}) holds for all $P\in\PC$.

  We now prove (\ref{eq:OptClaim1}). Note first that for all
  $\lambda>0$, $X\in\XC$, $Q\in\MC_V$ and $P\in\PC$,
  \begin{align*}
    E_P[U(X)]\leq V(\lambda Q|P)+\lambda E_Q[X]\leq V(\lambda Q|P).
  \end{align*}
  In particular,
  \begin{align*}
    \inf_{P\in\PC}E_P[U(X)]\leq
    \inf_{\lambda>0}\inf_{(Q,P)\in\MC_V}V(\lambda
    Q|P)\stackrel{\text{~(\ref{eq:Duality1})}}=u(0),\quad \forall
    X\in\XC,
  \end{align*}
  On the other hand, (\ref{eq:ProofAlimsup}) shows
  \begin{align*}
    u(0)=\lim_n\inf_{P\in\PC}E_P[U(\tilde k^n)]\leq
    \inf_{P\in\PC}\limsup_nE_P[U(\tilde k^n)]\leq \inf_{P\in\PC}E_P[U(\Xh)].
  \end{align*}
  This concludes the proof of (\ref{eq:OptClaim1}).

  We proceed to (\ref{eq:PredRep1}). Notice that
  \begin{align}\label{eq:ProofPredRepEq1}
    U(\Xh)=V(\hat\lambda d\Qh/d\Ph)+\hat\lambda (d\Qh/d\Ph)\Xh,\text{
      $\Ph$-a.s.}
  \end{align}
  Indeed, ``$\leq$'' is just a Young's inequality, while ``$\geq$''
  follows from
  \begin{align*}
    u(0)&= \inf_{P\in\PC}E_P[U(\Xh)]\leq E_{\Ph}[U(\Xh)]
    \stackrel{\text{(i)}}\leq E_{\Ph}\left[V\left(\hat\lambda
        \frac{d\Qh}{d\Ph}\right)+\hat\lambda
      \frac{d\Qh}{d\Ph}\Xh\right] \\
    &\stackrel{\text{(ii)}}\leq V(\hat\lambda
    \Qh|\Ph)\stackrel{\text{~(\ref{eq:Duality1})}}=u(0).
  \end{align*}
  Here (i) follows from the ``$\leq$'' part, and (ii) from
  $\Xh\in\XC$. 
  In particular, $\Ph$ attains the infimum in (\ref{eq:OptClaim1}) and
  we obtain (\ref{eq:ProofPredRepEq1}). But an elementary knowledge
  from convex analysis shows that this is possible only if
  \begin{align*}
    \Xh=-V'(\hat\lambda d\Qh/d\Ph),\text{ $\Ph$-a.s.}
  \end{align*}
  This is the first equality in (\ref{eq:PredRep1}), and the
  $\Qh$-a.s. uniqueness of $\Xh$ follows from that of $\hat\lambda
  d\Qh/d\Ph$ (see (\ref{eq:MaximalSol1})). On the other hand, the
  existence of $\hat\theta\in L(S,\Qh)$ with $\theta_0=0$ and
  $\hat\theta\cdot S$ being a $\Qh$-martingale, which represents
  $-V'(\hat\lambda d\Qh/d\Ph)$ as (\ref{eq:PredRep1}), follows from
  Theorem~3.2 of \citep{goll_ruschendorf01} (see also
  \citep[Theorem~2.2
  (iv)]{schachermayer2001:_optim_inves_in_incom_market}).  Finally,
  $\Qh$-a.s. uniqueness of the process $\hat\theta\cdot S$ follows
  from the $\Qh$-a.s. uniqueness of the terminal value
  $\hat\theta\cdot S_T$ and the fact that $\hat\theta\cdot S$ is a
  $\Qh$-martingale.
\end{proof}

\section{Uniform Supermartingale Property of Optimal Wealth}
\label{sec:UniSM}

We now proceed to the uniform supermartingale property of the optimal
wealth, that is, we shall show that $\hat\theta\cdot S$ is a
supermartingale under all local martingale measures $Q$ with finite
entropy w.r.t. some $P\in\PC$.  As outlined in Section~\ref{sec:Main},
this will follow if we can prove the dynamic variational inequality
(\ref{eq:DVI1}) for every $Q\in\MC_V$. Therefore, the key of this
section is the next proposition which should be compared with
\citep[][Lemma~3.12]{follmer_gundel06}. Recall that we have only to
consider the case $x=0$. \emph{In what follows, all the assumptions of
  Theorem~\ref{thm:MainThm1} are in force, and we do not cite them in
  each statement}.
\begin{prop}
  \label{prop:Tech} 
  We have
  \begin{enumerate}
  \item for all $Q\in\MC_V$, and for all stopping time $\tau\leq T$,
    \begin{equation}
      \label{eq:DynVarIneq1} %
      E_Q\left[ V'\left(\hat\lambda d\Qh/d\Ph\right)\Bigm|\FC_\tau
      \right]\geq E_{\Qh}\left[ V'\left(\hat\lambda
          d\Qh/d\Ph\right)\Bigm|\FC_\tau\right],\,Q\text{-a.s.}
    \end{equation}
  \item for all $P\in\PC$, and for all stopping time $\tau\leq T$,
    \begin{equation}
      \label{eq:VIcondP}
      E_P[U(\hat\theta\cdot S_T)|\FC_\tau]\geq E_{\Ph}[U(\hat\theta\cdot S_T)|\FC_\tau],\, P\text{-a.s.}
    \end{equation}

  \end{enumerate}

\end{prop}

We introduce some notations.  If $L$ is a strictly positive
martingale, we denote $L_{\tau,T}:=L_T/L_\tau$, for any stopping time
$\tau\leq T$. Recall that any probability $Q\ll\PB$ is identified with
a (uniformly integrable) martingale, namely the \emph{density process}
$Z^Q_\cdot =(dQ/d\PB)|_{\FC_\cdot}$.  In what follows, we denote by
$\Zh$ (resp. $\Dh$) the density process of $\Qh$ (resp. $\Ph$). Also,
when a pair $(Q,P)\in\MC_{loc}\times\PC$ is fixed, the density process
of $Q$ (resp. $P$) is denoted by $Z$ (resp.  $D$), and set:
\begin{equation}
  \label{eq:ConditionalConvexComb}
  Z^\alpha_{\tau,T}:=\alpha Z_{\tau,T}+(1-\alpha)\Zh_{\tau,T}, 
  \quad D^\alpha_{\tau,T}:=\alpha D_{\tau,T}+(1-\alpha)\Dh_{\tau,T},\quad \alpha \in[0,1].
\end{equation}

We make a couple of simple reductions. The first one is just a
notational reduction. In our purpose, we can assume without loss of
generality that $\hat\lambda=1$ \emph{since we already know
  $\hat\lambda$}. Indeed, $(\hat\lambda\Qh,\Ph)$ minimizes
$(\nu,P)\mapsto V(\nu|P)$ if and only if $(\Qh,\Ph)$ minimizes
$(\nu,P)\mapsto
V_{\hat\lambda}(\nu|P):=\frac{1}{\hat\lambda}V(\hat\lambda \nu|P)$.
Next, we have only to prove (\ref{eq:DynVarIneq1}) and
(\ref{eq:VIcondP}) for all $Q\in \MC_V^e$ and $P\in\PC^e$,
respectively. Indeed, if we could show (\ref{eq:DynVarIneq1}) for all
$Q'\in\MC_V^e$ for instance, we have $\bar Q:=(Q+\Qh)/2\in \MC_V^e$
for any $Q\in\MC_V$ on the one hand, and on the other hand, Bayes'
formula implies
\begin{align*}
  E_{\Qh}[\Phi|\FC_\tau]&\leq E_{\bar Q}[\Phi|\FC_\tau]\\
  &=\frac{Z_\tau}{Z_\tau+\Zh_\tau}E_Q[\Phi|\FC_\tau]+\frac{\Zh_\tau}{Z_\tau+\Zh_\tau}E_{\Qh}[\Phi|\FC_\tau]\,
  \text{a.s. on }\{Z_\tau>0\}
\end{align*}
where $\Phi=V'(d\Qh/d\Ph)$, hence (\ref{eq:DynVarIneq1}).  A similar
argument applies also to (\ref{eq:VIcondP}).

The first step is to show a ``Bellman-type'' principle for a
\emph{time-consistent} optimization. Note that the set $\MC_{loc}$ of
\emph{all} local martingale measures is m-stable, while $\MC_V$ is
not. The next simple lemma allows us to avoid this difficulty.
\begin{lem}
  \label{lem:ProofTech1}
  Let $(Q,P)\in \MC_V^e\times\PC^e$ with $V(Q|P)<\infty$, and $(Z,D)$
  the corresponding density processes as well as $\alpha\in
  [0.1]$. Then for any stopping time $\tau\leq T$, the random variable
  $\Dh_\tau D^\alpha_{\tau,T}V\left(\frac{\Zh_\tau
      Z^\alpha_{\tau,T}}{\Dh_\tau
      D^\alpha_{\tau,T}}\right)$ is $\FC_\tau$-locally
  integrable i.e., there exists an increasing sequence
  $A_n\in\FC_\tau$ such that
  \begin{align}\label{eq:LocInteg}
    \PB(A_n)\nearrow 1\quad\text{and}\quad 1_{A_n}\Dh_\tau
    D^\alpha_{\tau,T}V\left(\frac{\Zh_\tau
        Z^\alpha_{\tau,T}}{\Dh_\tau
        D^\alpha_{\tau,T}}\right) \in L^1,\,\forall n.
  \end{align}

\end{lem}
\begin{proof}
  Since $\Dh_\tau D_{\tau,T}^\alpha V\left(\frac{\Zh_\tau
      Z_{\tau,T}^\alpha}{\Dh_\tau D^\alpha_{\tau,T}}\right)\leq \alpha
  \Dh_\tau D_{\tau,T} V\left(\frac{\Zh_\tau Z_{\tau,T}}{\Dh_\tau
      D_{\tau,T}}\right)+(1-\alpha)
  \Dh_TV\left(\frac{\Zh_T}{\Dh_T}\right)$ (see the proof of
  Lemma~\ref{lem:ConseqRAE} below), and the second term is integrable,
  it suffices to prove the case $\alpha=1$.

  Recall from \citep{frittelli_rosazza04} that the condition
  (\ref{eq:RAE}) of reasonable asymptotic elasticity is equivalent to:
  for any $a\geq 1$, there exists $C_a,C_a'>0$ such that
  \begin{equation}
    \label{eq:REAequiv}
    V(\lambda y)\leq C_a V(y)+C_a'(y+1),\quad\forall \lambda\in[a^{-1},a],\,\forall y>0.
  \end{equation}
  Since $V$ is bounded from below by $U(0)$, we can choose the
  constant $C_a'$ so that the right hand side is always positive.  For
  the sequence $A_n$, we take
  \begin{align*}
    A_n:=\{\Zh_\tau,Z_\tau,\Dh_\tau,D_\tau \in(n^{-1},n)\}\in\FC_\tau,\quad\forall n.
  \end{align*}
  Noting that $\varphi:=\Dh_\tau D_{\tau,T}V\left(\frac{\Zh_\tau
      Z_{\tau,T}}{\Dh_\tau D_{\tau,T}}\right)=\frac{\Dh_\tau}{D_\tau}
  D_TV\left(\frac{\Zh_\tau D_\tau}{\Dh_\tau
      Z_\tau}\frac{Z_T}{D_T}\right)$, (\ref{eq:REAequiv}) implies that
  \begin{align*}
    \varphi\leq n^2 C_{n^4}D_TV(Z_T/D_T)+n^2 C_{n^4}'(Z_T+D_T) \text{
      a.s.  on }A_n.
  \end{align*}
  Thus $1_{A_n}\varphi\in L^1$ for each $n$. Finally,
  $\PB(A_n)\nearrow1$ since $\Qh\sim\Ph\sim Q\sim P\sim\PB$ by
  assumption.
\end{proof}

\begin{lem}
  \label{lem:DynVI1}
  For any $(Q,P)\simeq (Z,D)\in\MC^e_V\times\PC^e$ with
  $V(Q|P)<\infty$, $\alpha \in [0,1]$,
  \begin{equation}
    \label{eq:DynamicDual1}
    \begin{split}
      & E\left[\Zh_TV\left( \frac{\Zh_T}{\Dh_T}\right)
        \Bigm|\FC_\tau\right] \leq E\left[\Dh_\tau D^\alpha_{\tau,T}
        V\left( \frac{\Zh_\tau
            Z^\alpha_{\tau,T}}{\Dh_{\tau}D^\alpha_{\tau,T}}\right)
        \Bigm|\FC_\tau\right]\,\text{a.s.}
    \end{split}
  \end{equation}

\end{lem}

\begin{proof}
  Note first that the conditional expectation of the right hand side
  is well-defined and a.s. finite by Lemma~\ref{lem:ProofTech1}.  Let
  $C'$ be the set on which the inequality (\ref{eq:DynamicDual1})
  fails, which is $\FC_\tau$-measurable.  Then we suppose by way of
  contradiction that $\PB(C')>0$.

  Take a sequence $(A_n)\subset\FC_\tau$ as in
  Lemma~\ref{lem:ProofTech1} and a large $n$ so that $\PB(C'\cap
  A_n)>0$. Setting $C:=C'\cap A_n$, we define a new pair $(\bar Q,\bar
  P)\simeq (\bar Z,\bar D)$ by
  \begin{align*}
    \bar Z_T=1_{C^c}\Zh_T+1_C\Zh_\tau Z^\alpha_{\tau,T}\text{ and
    }\bar D_T=1_{C^c}\Dh_T+1_C\Dh_\tau D^\alpha_{\tau,T}.
  \end{align*}
  First, $(\bar Q, \bar P)\in\MC_{loc}\times\PC$ by the m-stability of
  $\MC_{loc}$ and $\PC$. Also, since
  \begin{align*}
    \bar D_TV\left(\frac{\bar Z_T}{\bar D_T}\right)
    =1_{C^c}\Dh_TV\left(\frac{\Zh_T}{\Dh_T}\right) +1_C\Dh_\tau
    D^\alpha_{\tau,T} V\left(\frac{\Zh_\tau
        Z^\alpha_{\tau,T}}{\Dh_\tau D^\alpha_{\tau,T}}\right),
  \end{align*}
  we have $V(\bar Q|\bar P)<\infty$ by the construction of $C$ and
  Lemma~\ref{lem:ProofTech1}, hence $\bar Q\in \MC_V$. Finally,
  \begin{align*}
    V(\bar Q|\bar P)&
    =E\left[\bar D_TV\left(\frac{\bar Z_T}{\bar D_T}\right)\right]\\
    =&E\Biggl[1_{C^c}E\left[\Dh_TV\left(\frac{\Zh_T}{\Dh_T}\right)\Bigm|\FC_\tau\right]
    +1_C E\left[\Dh_\tau D^\alpha_{\tau,T} V\left(\frac{\Zh_\tau
          Z^\alpha_{\tau,T}}{\Dh_\tau
          D^\alpha_{\tau,T}}\right)\Bigm|\FC_\tau\right]\Biggr]\\
    <&V(\Qh|\Ph).
  \end{align*}
  This contradict to the optimality of $(\Qh,\Ph)$. 
\end{proof}

Now the formal inequality in \textbf{Step~1} at the end of
Section~\ref{sec:Main} is realized as follows.
\begin{prop}
  \label{prop:AuxVI}
  For any $(Q,P)\simeq (Z,D)\in \MC_V^e\times\PC^e$ with
  $V(Q|P)<\infty$,
  \begin{equation}
    \label{eq:AUXVI1}
    \begin{split}
      \Zh_\tau&\left\{E_Q\left[V'\left(d\Qh/d\Ph\right)\Bigm|\FC_\tau\right]-E_{\Qh}\left[V'\left(d\Qh/d\Ph\right)\Bigm|\FC_\tau\right]\right\}\\
      &\quad +
      \Dh_\tau\left\{E_P[U(\Xh)|\FC_\tau]-E_{\Ph}[U(\Xh)|\FC_\tau]\right\}\geq
      0, \text{ a.s.}
\end{split}
\end{equation}

\end{prop}
\begin{proof}
  Let $(Z,D)$, $\tau$, $\alpha$ be as above, and set
\begin{align*}
  G_\tau(\alpha):=\Dh_\tau D_{\tau,T}^\alpha V( \Zh_\tau
  Z^\alpha_{\tau,T}/\Dh_\tau D^\alpha_{\tau,T}).
\end{align*}
Then $\alpha \mapsto G_\tau(\alpha)$ is convex (a.s.) by (the proof
of) Lemma~\ref{lem:ConseqRAE} below, hence
$(G_\tau(\alpha)-G(0))/\alpha$ decreases a.s. to the limit
$\Xi_\tau(Q,P)$ as $\alpha\searrow 0$. Here $\Xi_\tau(Q,P)$ is
explicitly computed as:
\begin{align*}
  \Xi_\tau(Q,P)&= \Zh_\tau V'\left( \frac{d\Qh}{d\Ph}\right)
  (Z_{\tau,T}-\Zh_{\tau,T}) +\Dh_\tau U(\Xh)
  (D_{\tau,T}-\Dh_{\tau,T}),
\end{align*}
using $\Zh_T/\Dh_T=d\Qh/d\Ph$ and $U(\Xh)=V(
d\Qh/d\Ph)-(d\Qh/d\Ph)V'(d\Qh/d\Ph)$. Since $G_\tau(1)$ is
$\FC_\tau$-locally integrable and
$E[(G_\tau(\alpha)-G_\tau(0))/\alpha|\FC_\tau]\geq 0$ a.s. by
Lemma~\ref{lem:DynVI1}, the (generalized) conditional monotone
convergence theorem shows that $E[\Xi(Q,P)|\FC_\tau]\geq 0$.  Noting
that $V'(d\Qh/d\Ph)=-\Xh\in L^1(Q)$ and $U(\Xh)\in L^1(P)$ by
Theorem~\ref{thm:Indirect}, we deduce (\ref{eq:AUXVI1}) from Bayes'
formula.
\end{proof}

We proceed to \textbf{Step~2}. Fixing $Q\in\MC_V$, we want to take $P$
``arbitrarily close'' to $\Ph$. The next simple lemma gives a precise
form of this argument. 

\begin{lem}
  \label{lem:ConseqRAE} 
  Let $(Q,P)$ and $(Q',P')$ be any two pairs of probability measures
  absolutely continuous w.r.t. $\PB$. Then for any $\alpha,\gamma\in
  (0,1)$, we have
  \begin{equation}
    \label{eq:RAEconseqEntropy}
    \begin{split}
      V(\alpha Q+&(1-\alpha)Q'|\gamma P+(1-\gamma)P')\\
      &\leq \gamma V\left(\frac\alpha\gamma
        Q\Bigm|P\right)+(1-\gamma)V\left(\frac{1-\alpha}{1-\gamma}Q'\Bigm|P'\right).
    \end{split}
  \end{equation}
  In particular, $V(Q|P)<\infty$ and $V(Q'|P')<\infty$ imply $V(\alpha
  Q+(1-\alpha)Q'|\gamma P+(1-\gamma)Q')<\infty$ for any
  $\alpha,\gamma\in (0,1)$.
\end{lem}
\begin{proof}
  Note that for any positive numbers $x,x',y,y'$,
  \begin{align*}
    \frac{\alpha x+(1-\alpha)x'}{\gamma y+(1-\gamma)y'}&=\frac{\gamma
      y}{\gamma y+(1-\gamma)y'}\frac{\alpha}{\gamma}\frac{x}{y}
    +\frac{(1-\gamma)y'}{\gamma
      y+(1-\gamma)y'}\frac{1-\alpha}{1-\gamma}\frac{x'}{y'}.
  \end{align*}
  Thus the convexity of $V$ shows that
  \begin{align*}
    (\gamma y+(1-\gamma)y')&V\left(\frac{\alpha x+(1-\alpha)x'}{\gamma
        y+(1-\gamma)y'}\right)\\
    & \leq\gamma yV\left(\frac\alpha\gamma
      \frac{x}{y}\right)+(1-\gamma)y'
    V\left(\frac{1-\alpha}{1-\gamma}\frac{x'}{y'}\right).
  \end{align*}
  Putting $dQ/d\PB$ (resp. $dQ'/d\PB$, $dP/d\PB$, $dP'/d\PB$) into $x$
  (resp. $x'$, $y$, $y'$), and taking the $\PB$-expectation, this
  implies (\ref{eq:RAEconseqEntropy}).  The second claim follows from
  the fact that $V(Q|P)<\infty$ $\Rightarrow$ $V(\lambda Q|P)<\infty$
  for any $\lambda>0$, as a consequence of reasonable asymptotic
  elasticity.
\end{proof}

\begin{proof}[Proof of Proposition~\ref{prop:Tech}]
  As noted after the statement of Proposition~\ref{prop:Tech}, we have
  only to consider the case $(Q,P)\in\MC_V^e\times\PC^e$ with
  $V(Q|P)<\infty$. Fixing such a pair $(Q,P)$, we put
  $Q_\alpha:=\alpha Q+(1-\alpha)\Qh$ and $P_\gamma:=\gamma
  P+(1-\gamma)\Ph$ for any $\alpha, \gamma \in (0,1)$.
By Lemma~\ref{lem:ConseqRAE}, the auxiliary variational inequality
(\ref{eq:AUXVI1}) is valid for any $(Q_\alpha,P_\gamma)$ with
arbitrary $\alpha,\gamma\in (0,1)$.  Noting that
$E_{Q_\alpha}[\Phi|\FC_\tau]-E_{\Qh}[\Phi|\FC_\tau]=\frac{\alpha
  Z_\tau}{\alpha
  Z_\tau+(1-\alpha)\Zh_\tau}\{E_Q[\Phi|\FC_\tau]-E_{\Qh}[\Phi|\FC_\tau]\}$
etc, we have
\begin{align*}
  &\Zh_\tau \frac{\alpha Z_\tau}{\alpha Z_\tau+(1-\alpha)\Zh_\tau}\{E_Q[V'(d\Qh/d\Ph)|\FC_\tau]-E_{\Qh}[V'(d\Qh/d\Ph)|\FC_\tau]\} \\
  &\qquad +\Dh_\tau\frac{\gamma D_\tau}{\gamma
    D_\tau+(1-\gamma)\Dh_\tau}
  \{E_P[U(\Xh)|\FC_\tau]-E_{\Ph}[U(\Xh)|\FC_\tau]\}\\
&\qquad\qquad \geq 0,\text{a.s. }\forall \alpha,\gamma \in (0,1).
\end{align*}
Since $\gamma D_\tau/(\gamma D_\tau
+(1-\gamma)\Dh_\tau)\stackrel{\gamma\downarrow 0}\rightarrow 0$ and
$\alpha Z_\tau/(\alpha Z_\tau
+(1-\alpha)\Zh_\tau)\stackrel{\alpha\downarrow 0}\rightarrow 0$, we
deduce (\ref{eq:DynVarIneq1}) and (\ref{eq:VIcondP}) by letting
$\gamma\downarrow 0$ (resp. $\alpha\downarrow 0$) with $\alpha$
(resp. $\gamma$) being fixed, whenever $V(Q|P)<\infty$. Finally, any
$Q\in\MC_V^e$ (resp. $P\in\PC^e$) admits a $P\in\PC$
(resp. $Q\in\MC_V$) with $V(Q|P)<\infty$ by definition (resp. by
Remark~\ref{rem:ConseqFinUtil}). 
\end{proof}

\begin{proof}[Proof of Theorem~\ref{thm:MainThm1}]
  Under the assumption $\Qh\sim\PB$, the $(S,\PB)$-integrability of
  $\hat\theta$ is clear.  We verify that $\hat\theta\cdot S$ is a
  supermartingale under each $Q\in\MC_V$. Since $V'(d\Qh/d\Ph)\in
  L^1(Q)$, the process defined by
  $M_\tau^Q=-E_Q[V'(d\Qh/d\Ph)|\FC_\tau]$ is a $Q$-martingale.  Then
  (\ref{eq:PredRep1}), (\ref{eq:DynVarIneq1}) as well as the fact that
  $\hat\theta\cdot S$ is a $\Qh$-martingale show that
  \begin{align*}
    \hat\theta\cdot S_\tau= -E_{\Qh}[V'(d\Qh/d\Ph)|\FC_\tau]\geq M^Q_\tau,\,Q\text{-a.s. }
  \end{align*}
  for any stopping time $\tau\leq T$. A stochastic integral w.r.t. a
  $Q$-local martingale dominated below by a $Q$-(uniformly integrable)
  martingale is a $Q$-supermartingale by
  \citep[][Theorem~1]{MR2053055}, which is a variant of
  Ansel-Stricker's lemma
  \citep[Proposition~3.3]{ansel_stricker94}.
\end{proof}

\section*{Acknowledgements}
Part of this research was carried out during the author's visit to ETH
Zürich. He warmly thanks Prof.  Martin Schweizer for a number of
valuable suggestions and his hospitality. The author gratefully
acknowledges the financial support from The Norinchukin Bank and the
Global COE program ``The research and training center for new
development in mathematics''.


\end{document}